\documentclass[a4paper,UKenglish,autoref]{lipics-v2019}

\usepackage[ruled,vlined,linesnumbered,commentsnumbered]{algorithm2e}
\usepackage{amsmath,amssymb,amsfonts}

\title{On {\em k}-connectivity oracles in {\em k}-connected graphs}
\titlerunning{On {\em k}-connectivity oracles}

\author{Zeev Nutov}{The Open University of Israel}{nutov@openu.ac.il}
{https://orcid.org/0000-0002-6629-3243}{}
\authorrunning{Zeev Nutov}
\Copyright{Zeev Nutov}

\nolinenumbers


\acknowledgements{}%optional

\begin{document}

\maketitle

% Sets
\newcommand{\sem}    {\setminus}
\newcommand{\subs}   {\subseteq}
\newcommand{\empt}  {\emptyset}

\newcommand{\f}   {\frac}

% Calliraphic
\newcommand{\Aa}  {{\cal A}}
\newcommand{\Bb}  {{\cal B}}
\newcommand{\CC}  {{\cal C}}
\newcommand{\DD}  {{\cal D}}
\newcommand{\FF}  {{\cal F}}
\newcommand{\HH}  {{\cal H}}
\newcommand{\PP}  {{\cal P}}
\newcommand{\RR}  {{\cal R}}
\newcommand{\Sa}   {{\cal S}}
\newcommand{\TT}   {{\cal T}}
\newcommand{\LL}   {{\cal L}}

% Queries
\newcommand{\con}   {{\sc pcon$_k(s,t)$}} 
\newcommand{\cut}     {{\sc pcut$_k(s,t)$}} 
\newcommand{\econ} {{\sc con$_k(s,t)$}}
\newcommand{\mcut}  {{\sc cut$_k(s,t)$}}  

% Greek
\newcommand{\al}   {\alpha}
\newcommand{\be}   {\beta}
\newcommand{\ga}   {\gamma}
\newcommand{\de}   {\delta}
\newcommand{\ka}   {\kappa}
\newcommand{\la}   {\lambda}
\newcommand{\eps}   {\epsilon}

\newcommand{\Om}   {\Omega}

% Bisets
\newcommand{\A}   {\mathbb{A}}
\newcommand{\B}   {\mathbb{B}}
\newcommand{\C}   {\mathbb{C}}
\newcommand{\p}   {\partial}

\def\lca   {{\sf lca}} 
\def\dist  {{\sf dist}} 
\def\o     {{\sf o}}
\newcommand{\Ga}  {\Gamma}

\newcommand {\ignore} [1] {}

%%%%%%%
%%%%%%%

\begin{abstract}
A {\bf $k$-connectivity oracle} for a graph $G=(V,E)$ is a data structure that 
given $s,t \in V$ determines whether there are at least $k+1$ internally disjoint $st$-paths in $G$. 
For undirected graphs, Pettie, Saranurak \& Yin \cite{PSY} proved that 
any $k$-connectivity oracle requires $\Omega(kn)$ bits of space.
% They asked whether $\Omega(kn)$ bits are still necessary if $G$ is $k$-connected,
% and conjectured that a $k$-connectivity oracle for a $k$-connected graph requires $\tilde{O}(n)$ bits of space. 
We will show by a very simple proof that this is so even if $G$ is $k$-connected, 
answering an open question posed in \cite{PSY}. 
\end{abstract}

%%%%%%%%%%%%%%%%%%%%
\section{Introduction} \label{s:intro}
%%%%%%%%%%%%%%%%%%%%

Given a directed or undirected graph $G=(V,E)$ let 
$\ka(s,t)$ denote the maximum number of internally-disjoint $st$-paths in $G$;
$G$ is {\bf $k$-connected} if $\ka(s,t) \geq k$ for all $(s,t) \in V \times V$.
A {\bf $k$-connectivity oracle} for $G$ is a data structure that  
given $s,t \in V$ determines whether  $\ka(s,t) > k$.
We will consider the following question -- what is the minimum number of bits that 
a $k$-connectivity oracle of an $n$-node graph needs?

It is not hard to establish that for both directed and undirected graphs 
at least $n$ bits are needed even for $k=0$. 
For undirected graphs and $k=0$ we just need to give a different label for every connected component,
and thus $O(n \log n)$ bits suffice.
More generally, in the case of {\em edge-connectivity}, 
the relation ``there are $k$-edge-disjoint $st$-paths'' is an equivalence, 
which implies that also in this case $O(n \log n)$ bits suffice. 
This bound is also achievable by the Gomory-Hu cut-tree \cite{GH}, 
that can also return a minimum $st$-edge-cut.
In contrast, at least $n^2/4$ bits are needed to answer reachability queries (the case $k=0$)  
for a directed acyclic graph (DAG), c.f. \cite{KR-posets,MN,DGJ};
this follows from the fact that the number of distinct posets on $n$ elements 
is at least $2^{n^2/4}$ \cite{KR-posets}.

Now let us consider node-connectivity in undirected graphs.
It is well known that any graph $G$ has for any $k$ has 
a ``sparse certificate'' -- a spanning subgraph with $O(kn)$ edges,  
that preserves pairwise connectivities (and also pairwise minimum node-cuts) 
of $G$ up to $k$, see \cite{NI}. 
Such a subgraph can serve as a $k$-connectivity oracle.
Pettie, Saranurak \& Yin \cite{PSY} proved that this is essentially the best possible -- 
a $k$-connectivity oracle for an undirected graph needs $\Omega(kn)$ bits of space.
They asked whether $\Omega(kn)$ bits are still necessary if $G$ is $k$-connected.
%  and conjectured that then $\tilde{O}(n)$ bits will suffice.
We will answer this question by proving the following. 

\begin{theorem} \label{t1}
For any integers $k \ge 1$ and $1 \le p \le 2^{k^2/4}$ there exists 
a $k$-connected graph $G=G(k,p)$ with $n=k(p+1)$ nodes 
such that any $k$-connectivity oracle for $G$ requires at least $p \cdot k^2/4 = \left(1-\f{1}{p+1}\right) \cdot kn/4$ bits of space.
Consequently, a $k$-connectivity oracle for a $k$-connected graph requires $\Omega(kn)$ bits of space.
\end{theorem}

% $n=k(p+1) \ \ \ \ \ kp=n-k=n-n/(p+1)=n[1-1/(p+1)]$

For $p=1$ we get $k=n/2$ while for $p=2^{k^2/4}$ we get $k=2\sqrt{\lg p} \approx 2 \sqrt{\lg n}$. 
Hence in terms of $n$ our construction applies (roughly) in the range $2\sqrt{\lg n} \le k \le n/2$.
We also note that our construction is substantially simpler than the one in \cite{PSY}.

Since distinguishing between $\ka(s,t)=k$ and $\ka(s,t)>k$ requires $\Omega(kn)$ bits of space
even for $k$-connected graphs, our construction also rules out existence of compact approximate 
$k$-connectivity oracle, see \cite{PSY} for details.  

For various $k$-connectivity oracles for general and $k$-connected graphs see \cite{HL, IN, PY, N-ds, PSY}.

%%%%%%%%%%%%%%%%%%%%%%%%%%%
\section{Proof of Theorem~\ref{t1}} \label{s2}
%%%%%%%%%%%%%%%%%%%%%%%%%%%

We will use a construction of \cite{LN}, that for any digraph $D=(R,A)$ constructs 
an undirected graph $G=(R \cup R',E)$ (where $R'$ is a copy of $R$) that ``shifts'' the connectivities by exactly $|R|$. 
Formally, this undirected graph is defined as follows. 

\begin{definition}
The {\bf connectivity-shift graph} of a directed graph $D=(R,A)$ is an undirected graph $G=(V,E)$ obtained as follows.
\begin{enumerate}
\item
Add a copy $R'$ of $R$ and 
replace every arc $uv \in A$ by the undirected edge $uv'$, where $v' \in R'$ denotes the copy of $v \in R$.
\item
Add a clique on each of $R,R'$ and the matching $M=\{vv':v \in R\}$.
\end{enumerate}
\end{definition}

Let $\ka(G)=\min\{\ka(s,t):(s,t) \in V \times V\}$ denote the (node-)connectivity 
of a directed or an undirected graph $G$.
The following was proved in \cite{LN}.

\begin{lemma} [\cite{LN}] \label{l:LN}
Let $G$ be the connectivity-shift graph of a digraph $D=(R,A)$. Then the following holds.
\begin{enumerate}[(i)]
\item
$\ka_G(u,v')=\ka_D(u,v)+|R|$ for all $u,v \in R$.
\item
$\ka(G)=\ka(D)+|R|$. 
\end{enumerate}
\end{lemma}

As was mentioned, supporting reachability queries in an $n$-node DAG requires 
at least $n^2/4$ bits of space, c.f \cite{KR-posets,MN,DGJ}.
Note that $\ka(D)=0$ for any DAG $D$.
Combining with Lemma~\ref{l:LN} we get the following.

\begin{corollary} \label{c:shift}
Let $G$ be the connectivity-shift graph of a $k$-node DAG $D=(R,A)$. Then:
\begin{enumerate}
\item
$\ka_G(u,v')=\ka_D(u,v)+k$ for all $u,v \in R$.
\item
$G$ has $2k$ nodes and $\ka(G)=k$.
\item
Supporting $k$-connectivity queries in $G$ requires at least $k^2/4$ bits of space.
\end{enumerate}
\end{corollary}

Corollary~\ref{c:shift} already implies that
a data structure for answering $k$-connectivity queries 
for a $k$-connected graph with $n=2k$ nodes requires at least $k^2/4=kn/8$ bits of space.
We would like to show that this is so also when $n$ is much larger than $k$. 
A common technique to achieve this is just to take $p$ ``gadgets'' as in Corollary~\ref{c:shift}, 
of different DAGs that correspond to different posets; it is known that there are at least 
$2^{k^2/4}$ distinct posets on $k$ elements \cite{KR-posets}.
Then the graph will have $n=2pk$ nodes and the lower bound will be $pk^2/4 = n k/8$ bits.
This technique works for general graphs, namely, we get a simple proof that the lower bound of $kn/8$ bits 
applies for general graphs with $n=2pk$ nodes.
But it does not apply directly to $k$-connected graphs, as the obtained graph is not even connected.

To adopt this construction for $k$-connected graphs we modify it as follows. 
Let the $p$ gadget graphs have node sets $R \cup R_1, R \cup R_2, \ldots, R \cup R_p$.
We identify the set $R$ of all $p$ gadgets. 
The obtained graph $G$ has $n=k(p+1)$ nodes and a lower bound of 
$pk^2/4 = \frac{1-1/p}{4} \cdot kn$ space bits for $k$-connectivity oracle.
To finish the proof of Theorem~\ref{t1} we prove the following.

\begin{lemma}
The constructed graph $G$ is $k$-connected.
\end{lemma}
\begin{proof}
It is known that a graph $G$ is not $k$-connected if and only if 
there exists a proper partition $S,C,T$ of the node set of $G$ such that $|C| <k$ 
and such that $G$ has no $ST$-edge.
Suppose to the contrary that the constructed graph $G$ has such a partition.
If $S \cap R \ne \empt$ then $T \cap R=\empt$ and $T \cap R_i=\empt$ for all $i$, 
since for all $i$, the subgraph $G[R \cup R_i]$ induced by $R \cup R_i$ is $k$-connected;
this contradicts that $T \neq \empt$. 
Similarly, we cannot have $T \cap R \neq \emptyset$.
It follows therefore that $R \subs C$, giving the contradiction $|C| \geq |R|=k$. 
\end{proof}

This concludes the proof of Theorem~\ref{t1}. 

% \bibliographystyle{plainurl}
% \bibliography{ds-kc}

\end{document}